\documentclass[acmsmall,screen]{acmart}

\pdfoutput=1  % for arxiv.org

\setcopyright{acmlicensed}
\copyrightyear{2024}
\acmYear{2024}
\acmDOI{XXXXXXX.XXXXXXX}

\acmVolume{XXX}
\acmNumber{YYY}
\acmArticle{111}
\acmMonth{1}

\newtheorem{theorem}{Theorem}[section]

\newtheorem{corollary}[theorem]{Corollary}

\theoremstyle{definition}
\newtheorem{definition}{Definition}[section]

\theoremstyle{remark}
\newtheorem{remark}{Remark}[section]

\newtheorem{theoremsection}{Theorem}[section]

\newtheorem{lemma}[theoremsection]{Lemma}

\def\QEDopen{{\setlength{\fboxsep}{0pt}\setlength{\fboxrule}{0.2pt}\fbox{\rule[0pt]{0pt}{1.3ex}\rule[0pt]{1.3ex}{0pt}}}}
\def\QED{\QEDopen}

\def\Q.E.D{\hfill\QED}

\begin{document} 

\title{The Separation of $\mathcal{NP}$ and $\mathcal{PSPACE}$}%
%\titlenote{A conference version of this paper appeared in SODA 2021~\cite{mathieu2021competitive}.
%The journal version was published in ACM Transactions on Algorithms,
%\url{https://doi.org/10.1145/3672614} .}

 \date{\today}

\author{Tianrong Lin}
\affiliation{%
 \institution{Hakka University}
 %\city{}
 %\state{}
 \country{China}
}
\orcid{0000-0002-1187-2395}
%\email{}

% \renewcommand{\shortauthors}{Mathieu, Rajaraman, Young, Yousefi}

% \maketitle

\begin{abstract}
There is an important and interesting open question in computational complexity on the relation between the complexity classes $\mathcal{NP}$ and $\mathcal{PSPACE}$. It is a widespread belief that $\mathcal{NP}\ne\mathcal{PSPACE}$. In this paper, we confirm this conjecture affirmatively by showing that there is a language $L_d$ accepted by no polynomial-time nondeterministic Turing machines but accepted by a nondeterministic Turing machine running within space $O(n^k)$ for all $k\in\mathbb{N}_1$. We achieve this by virtue of the prerequisite of 
$$
{\rm NTIME}[S(n)]\subseteq{\rm DSPACE}[S(n)],
$$
and then by diagonalization against all polynomial-time nondeterministic Turing machines via a universal nondeterministic Turing machine $M_0$. We further show that $L_d\in \mathcal{PSPACE}$, which leads to the conclusion 
$$
\mathcal{NP}\subsetneqq\mathcal{PSPACE}.
$$

Our approach is based on standard diagonalization and novel new techniques developed in the author's recent works \cite{Lin21a,Lin21b} with some new refinement. 
\end{abstract}




\received{7 September 2026}
\received[revised]{7 September 2026}
\received[accepted]{7 September 2026}

\setcopyright{rightsretained}
%\acmJournal{JACM}
\acmYear{2026} \acmVolume{1} \acmNumber{1} \acmArticle{1} \acmMonth{1}\acmDOI{10.1145/xxxxxxx}

\maketitle
\vskip 0.3 cm
\section{Introduction}
\label{sec:introduction}
\vskip 0.3cm

{\it Computational complexity theory} has developed rapidly over the past $40$ years. {\em Space complexity}, one of the two primary metrics for evaluating a computation, was introduced by Stearns et al. \cite{SHL65} in 1965. Together with {\em time complexity}, it remains one of the most important and widely used measures of the complexity of a computation \cite{Pap94, AB09, Sip13}. In the field of computational complexity theory, numerous abstract models of computation have been introduced, the most widespread of which is the Turing machine; many of these models give rise to essentially equivalent formalizations of fundamental notions such as time and space complexity. For further details on space complexity, we refer the reader to the excellent survey by Michel \cite{Mic92}. Furthermore, the textbook \cite{Wig07} provides an in-depth explanation of the relationship between $\mathcal{P}$, $\mathcal{NP}$, and mathematics from a complexity-theoretic perspective.

In {\it computational complexity theory}, $\mathcal{NP}$ ({\em nondeterministic polynomial time}) is a complexity class used to classify decision problems (see e.g. \cite{ANMOUS2}), i.e., the set of decision problems for which the problem instances where the answer is ``yes" have proofs verifiable in polynomial time by a deterministic Turing machine, or alternatively the set of problems that can be solved in polynomial time by a nondeterministic Turing machine. In other words, an equivalent definition of $\mathcal{NP}$ is the set of decision problems verifiable in polynomial time by a deterministic Turing machine \cite{Kar72}. The reason these two definitions are equivalent is that the algorithm based on the Turing machine consists of two phases, the first of which consists of a guess about the solution, which is generated in a nondeterministic way, while the second phase consists of a deterministic algorithm that verifies whether the guess is a solution to the problem, see e.g. \cite{ANMOUS2}. 

Equally importantly, the complexity class $\mathcal{PSPACE}$ is the set of all decision problems that can be solved by a Turing machine using a {\em polynomial amount of space} (see e.g. \cite{ANMOUS3}). There are many other complexity classes, such as $L$,\footnote{ The notation $L$ will represent the complexity class ${\rm DSPACE}[\log n]$ in this paper unless otherwise stated.} $NL$, and $\mathcal{P}$, in addition to $\mathcal{NP}$ and $\mathcal{PSPACE}$, appearing in the modern textbooks of computational complexity; for example, see \cite{Pap94,AB09,Sip13}. For these complexity classes $L$, $NL$, $\mathcal{P}$, $\mathcal{NP}$, and $\mathcal{PSPACE}$, there is a well-known tower of inclusions stated below
$$
L\subseteq NL\subseteq \mathcal{P}\subseteq \mathcal{NP}\subseteq \mathcal{PSPACE}.
$$
As observed by Cook \cite{Coo00}, a simple diagonal argument shows that the first is a proper subset of the last, but we cannot prove that any particular adjacent inclusion is proper. To the best of our knowledge, it is also an important open question whether $\mathcal{NP}\overset{?}{=}\mathcal{PSPACE}$ in computational complexity theory, see e.g. \cite{Pap94,AB09,Sip13}, also \cite{ANMOUS1}. In particular, this problem was considered an obvious open problem that is likely at least as hard as some classical problems of mathematics (such as Fermat’s conjecture or the four-color problem\footnote{However, both of these mathematical problems have already been solved.}), in view of the amount of work that has been done searching for polynomial-time algorithms for $\mathcal{NP}$-complete problems; see, e.g., Research Problem 10.30 in \cite[p. 403]{AHU74}.

Although all problems solvable in polynomial time can be solved in polynomial space, it is still an unresolved question whether there exist problems solvable in polynomial space that cannot be solved in polynomial time \cite{GJ79}. We here focus on an even stronger question: {\em whether there exist problems solvable in polynomial space that cannot be solved in nondeterministic polynomial time}. We say the later question is even stronger than the former  because an affirmative answer to the later  implies a positive answer to the former . In other words, we study the relationship between $\mathcal{NP}$ and $\mathcal{PSPACE}$. Of course, one will ask, is $\mathcal{NP}$ equal to $\mathcal{PSPACE}$? With regard to this question, there are some previous works that try to answer it. For example, in 1981, Book argued in \cite{Boo81} that 
$$
\mathcal{NP} = \mathcal{PSPACE}
$$
if and only if $NP(A)$ is equal to the class $NPQUERY(A)$ of languages accepted by nondeterministic polynomial-space-bounded oracle machines that can query the oracle for $A$ only a polynomial number of times. 

In fact, there are widespread briefs that 
$$
\mathcal{NP}\neq \mathcal{PSPACE}, 
$$
but no proof is available in contemporary computational complexity textbooks such as \cite{Pap94,AB09, Sip13}. In this paper, we resolve this question. The reader will note that the essential techniques employed here are essentially the same as those in the author’s recent works \cite{Lin21a, Lin21b}, but with some refinements and additional technical ingredients, such as the deterministic simulation of nondeterministic Turing machines. Furthermore, we have already shown in \cite{Lin21b} that a contradictory relativization result is not a necessary and sufficient condition for $\mathcal{NP}\ne\mathcal{PSPACE}$ (or for $\mathcal{NP}=\mathcal{PSPACE}$); accordingly, we do not discuss in detail how to break various barriers such as those in \cite{AW09,BGS75}. The reader is referred to \cite{Lin21b} for a detailed and clear explanation, as even if we were to devote space to discussing this point here, the conclusion reached would be the same.

Indeed, in what follows, we first enumerate all of the polynomial-time nondeterministic Turing machines and then diagonalize against all of them, followed by a careful analysis. The standard diagonalization technique allows us to construct a language that differs from any language in $\mathcal{NP}$, which is similar to {\em time hierarchies} \cite{HS65,Coo73,Pap94,FS07, AB09, Sip13}, i.e., allowing us to produce a language that is at least not in a specific complexity class. The reader is referred to \cite{Tur37} or the survey article \cite{For00} for more detailed information about the standard diagonalization technique. 

In the following, let us prove step by step the relationship:
$$
\mathcal{NP}\subsetneqq\mathcal{PSPACE}
$$
between the complexity classes $\mathcal{NP}$ and $\mathcal{PSPACE}$.

\begin{theorem}
\label{theorem1}
   $\mathcal{NP}\neq \mathcal{PSPACE}$. In other words, $\mathcal{NP}\subsetneqq \mathcal{PSPACE}$.
\end{theorem}
  
The class $\mathcal{IP}$ of languages that have efficient interactive proofs of membership was introduced by Goldwasser et al. \cite{GMR85}. A language $L\in\mathcal{IP}$ if a probabilistic polynomial-time verifier $V$ can be convinced by some prover $P$ to accept any $x\in L$ with overwhelming probability but cannot be convinced by any prover $P'$ to accept any $x\not\in L$ with nonnegligible probability. A major result of complexity theory is that $\mathcal{PSPACE}$ can be characterized as all the languages recognizable by a particular interactive proof system (see e.g. \cite{ANMOUS2}), the one defining the class $\mathcal{IP}$. In this system, there is an all-powerful prover trying to convince a randomized polynomial-time verifier that a string is in the language. It should be able to convince the verifier with high probability if the string is in the language but should not be able to convince it except with low probability if the string is not in the language. More precisely, the above characterization shown by Shamir \cite{Sha92} is as follows:
$$
\mathcal{IP}=\mathcal{PSPACE}.
$$
  
Hence, an immediate corollary from Theorem \ref{theorem1} and Shamir's main result in \cite{Sha92} is as follows:
\begin{corollary}
\label{corollary1}
$$
\mathcal{NP}\subsetneqq\mathcal{IP}.
$$
\end{corollary}

\vskip 0.3cm
\subsection{Organization}
\vskip 0.3cm

The rest of this work is organized as follows. For the convenience of the reader, in the next section (Section \ref{sec:preliminaries}), we will recall some notions and notation closely associated with our discussion appearing in Section \ref{sec:introduction}. In Section \ref{sec:proofofmainresult}, we will prove our main result, i.e., Theorem \ref{theorem1}. Finally, a brief conclusion is drawn in the last section (Section \ref{sec:conclusions}), in which possible future study is also pointed out.

\vskip 0.3cm
\section{Preliminaries}
\label{sec:preliminaries}
\vskip 0.3cm
In this section, we describe some notation and notions that are closely associated with our discussion.

For a set $S$, we let $|S|$ denote the cardinality of $S$.

Let $\mathbb{N}$ denote the set of natural numbers
$$
\{0,1,2,3,\cdots\}
$$
where $+\infty\not\in\mathbb{N}$. Further, $\mathbb{N}_1$ denotes the set of positive integers 
$$
\mathbb{N}\setminus\{0\}.
$$
It is clear that there is a bijection between $\mathbb{N}$ and $\mathbb{N}_1$. To see this, simply let the bijection be 
$$
n \mapsto n+1
$$
where $n\in\mathbb{N}$ and $n+1\in\mathbb{N}_1$.

Let $\Sigma$ be an alphabet. For finite words $w\in\Sigma^*$, the length of $w$, denoted $|w|$, is defined to be the number of symbols in it. For a finite word $w\in\Sigma^*$ and positive integer $1\leq j\leq |w|$, $w_j$ denotes the $j$th symbol in $w$. For example, let $w=abcdeaa\in\Sigma^*$, then $w_1=a$, $w_2=b$, $w_3=c$, and so on.

The big $O$ notation indicates the order of growth of some quantity as a function of $n$ or the limiting behavior of a function. For example, that $S(n)$ is big $O$ of $f(n)$, i.e.,
$$
S(n)=O(f(n))
$$
means that there exists a positive integer $N_0$ and a positive constant $M$ such that
$$
S(n)\leq M\times f(n)
$$
for all $n>N_0$.

Throughout this paper, the computational modes used are {\it nondeterministic Turing machines}. We follow the standard definition of a nondeterministic Turing machine given in the standard textbook \cite{AHU74}. Let us first introduce the precise definition of a nondeterministic Turing machine as follows:

\begin{definition}[$k$-tape nondeterministic Turing machine, \cite{AHU74}]
\label{definition2.1}
A $k$-tape nondeterministic Turing machine (shortly, NTM) $M$ is a seven-tuple $(Q,T,I,\delta,\mathrm{b},q_0,q_f)$
where:
\begin{enumerate}
\item {$Q$ is the set of states.}
\item {$T$ is the set of tape symbols.}
\item {$I$ is the set of input symbols; $I\subseteq T$.}
\item {$\mathrm{b}\in T-I$, is the blank.}
\item {$q_0$ is the initial state.}
\item {$q_f$ is the final (or accepting) state.}
\item {$\delta$ is the next-move function, or a mapping from $Q\times T^k$ to subsets of 
$$
Q\times(T\times\{L,R,S\})^k.
$$
Suppose
$$
  \delta(q,a_1,a_2,\cdots,a_k)=\{(q_1,(a^1_1,d^1_1),(a^1_2,d^1_2),\cdots,(a^1_k,d^1_k)),\cdots,(q_n,(a^n_1,d^n_1),(a^n_2,d^n_2),\cdots,(a^n_k,d^n_k))\}
$$
and the nondeterministic Turing machine is in state $q$ with the $i$th tape head scanning tape symbol $a_i$ for $1\leq i\leq k$. Then in one move the nondeterministic Turing machine enters state $q_j$, changes symbol $a_i$ to $a^j_i$, and moves the $i$th tape head in the direction $d^j_i$ for $1\leq i\leq k$ and $1\le j\le n$.}
\end{enumerate}
\end{definition}

If $|\delta(q,a_1,a_2,\cdots,a_k)|\leq 1$ for all $q$ and for all $a_i$, $1\leq i\leq k$, then $M$ is deterministic, and in this case, we call $M$ a deterministic Turing machine.

Let $M$ be a nondeterministic Turing machine and $w$  an input. Then $M(w)$ represents that $M$ is on input $w$. 

A nondeterministic Turing machine $M$ works in time $t(n)$ (or is of time complexity $t(n)$) if for any input $w\in I^*$, where $I$ is the input alphabet of $M$, $M(w)$ will halt within $t(|w|)$ steps. We should make a formal definition of a polynomial-time nondeterministic Turing machine. In fact, the notion of a polynomial-time nondeterministic Turing machine can be defined in a similar way to the concept of a polynomial-time deterministic Turing machine \cite{Coo00}. In short, we have the following:

\begin{definition}[cf. polynomial-time deterministic Turing machines in \cite{Coo00}]
\label{definition2.2}
Formally, a polynomial-time nondeterministic Turing machine is a nondeterministic Turing machine such that there exists $k\in\mathbb{N}_1$, for all input $w$ of length $|w|$ (where $|w|\in\mathbb{N}$), $M(w)$ will halt within $t(|w|)=|w|^k+k$ steps. The language accepted by polynomial-time nondeterministic Turing machine $M$ is denoted by $L(M)$.
\end{definition}

\vskip 0.3 cm
\begin{remark}
\label{remark2.1}
Obviously, in the above definition, for any $i\in\mathbb{N}$, given a polynomial-time (say, $n^k+k$ time-bounded) nondeterministic Turing machine $M$, $M (x)$ will halt within $O(|x|^{k+i})$ steps for any input $x$ of length $|x|$. However, there exists at least an input $y$ of length $|y|$ such that $M(y)$ does not halt within $O(|y|^{k-1})$ steps.
\end{remark}
\vskip 0.3 cm

By default, a word $w$ is accepted by a polynomial-time (say, $t(n)$ time-bounded) nondeterministic Turing machine $M$ if there exists at least one computation path $\gamma$ putting the machine into an accepting state (i.e., stopping in the ``accepting" state) on input $w$ (in such a situation, we say that $M$ accepts $w$). Specifically, if $M$ accepts the language $L(M)$, then $w\in L(M)$ if and only if $M$ accepts $w$. In other words, there exists at least one accepting path for $M$ on input $w$. 

Let $M$ be a Turing machine. For every input $x$ of length $n$, if all computations of $M$ on $x$ take less than $T(n)$ steps, then $M$ is said to be of {\em time complexity} $T(n)$ or a {\em deterministic/nondeterministic $T(n)$ time-bounded Turing machine}. The class of languages recognized by some deterministic (resp. nondeterministic) Turing machine of time complexity $T(n)$ is denoted by ${\rm DTIME} [T(n)]$ (resp. ${\rm NTIME}[T(n)]$. The notation $\mathcal{NP}$ is defined to be the class of languages:
$$
\mathcal{NP}=\bigcup_{k\in\mathbb{N}_1}{\rm NTIME}[n^k].
$$

Let $M$ be a Turing machine. For every input $x$ of length $n$, if every computation of $M$ on $x$ scans at most $S(n)$ tape cells of each of its work tapes, then $M$ is said to be {\em of space complexity} $S(n)$ or to be a {\em deterministic/nondeterministic $S(n)$ space-bounded Turing machine}. The class of languages recognized by the deterministic (resp. nondeterministic) Turing machine of space complexity $S(n)$ is denoted by ${\rm DSPACE}[S(n)]$ (resp. ${\rm NSPACE} [S(n)]$). 

The complexity classes such as ${\rm DSPACE}[S(n)]$, ${\rm NSPACE}[S(n)]$,  $L$, and $NL$ are called space complexity classes. In particular, we use the notation $L$ (resp. $NL$) to denote the class ${\rm DSPACE}[\log n]$ (resp. ${\rm NSPACE}[\log n]$). The classes $\mathcal{PSPACE}$ and $\mathcal{NSPACE}$ are given by
$$
\mathcal{PSPACE}=\bigcup_{k\in\mathbb{N}_1}{\rm DSPACE}[n^k]
$$
and
$$
\mathcal{NSPACE}=\bigcup_{k\in\mathbb{N}_1}{\rm NSPACE}[n^k].
$$

By Savitch's Theorem \cite{Sav70}, we have 
$$
\mathcal{PSPACE}=\mathcal{NSPACE}.
$$

The following lemma is important and useful for our proof of the main result. For convenience, we quote its proof as follows:

\begin{lemma}[Theorem 7.4 (b) in \cite{Pap94}, p. 147]
\label{lemma1}
Suppose that $f(n)$ is a proper complexity function. Then ${\rm NTIME}[f(n)]\subseteq {\rm DSPACE}[f(n)]$.
\end{lemma}
\begin{proof}
Consider a language 
$$
L\in {\rm NTIME}[f(n)].
$$
There is a precise nondeterministic Turing machine $M_0$ that decides $L$ in time $f(n)$. We shall design a deterministic machine $M_1$ that decides $L$ in space $f(n)$. The deterministic machine $M_1$ generates a sequence of nondeterministic choices for $M_0$, that is, an $f(n)$-long sequence of integers between $0$ and $d-1$ (where $d$ is the maximum number of choices for any state-symbol combination of $M_0$). Then $M_1$ simulates the operation of $M_0$ with the given choices. This simulation can obviously be carried out in space $f(n)$ (in time $f(n)$, only $O(f(n))$ characters can be written!). However, there are exponentially many such simulations that must be tried, to check whether a sequence of choices that leads to acceptance exists. The point is that they can be carried out one-by-one, always erasing the previous simulation to reuse space. We only need to keep track of the sequence of choices currently simulated, and generate the next, but both tasks can easily be done within space $O(f(n))$. The fact that $f$ is proper can be used to generate the first sequence of choices, $0^{f(n)}$.
\end{proof}

Because one of our research objects is polynomial-time nondeterministic Turing machines, we will highlight two of their key properties:

\begin{lemma}[Lemma 10.1 in \cite{AHU74}]
\label{lemma3} 
If $L$ is accepted by a $k$-tape nondeterministic Turing machine of time complexity $T(n)$, then $L$ is accepted by a single-tape nondeterministic Turing machine of time complexity $O(T^2(n))$. \Q.E.D
\end{lemma}

\begin{lemma}[Corollary 2 in \cite{AHU74}, p. 372]
\label{lemma4}
If $L$ is accepted by a $k$-tape nondeterministic Turing machine of space complexity $S(n)$, then $L$ is accepted by a single-tape nondeterministic Turing machine of space complexity $S(n)$, where $S(n)\geq n$.\Q.E.D
\end{lemma}

Finally, the relationship between $\mathcal{NP}$ and $\mathcal{PSPACE}$ is to ask whether 
$$
\mathcal{NP}\overset{?}{=}\mathcal{PSPACE},
$$
i.e., whether the classes of languages $\mathcal{NP}$ and $\mathcal{PSPACE}$ are identical.

\vskip 0.3cm
\section{Enumeration of Polynomial-Time Nondeterministic Turing Machines}
\label{sec:enumeration}
\vskip 0.3cm

By Definition \ref{definition2.2}, a polynomial-time nondeterministic Turing machine can be represented by a tuple of $(M, k)$, where $M$ is the nondeterministic Turing machine itself and $k$ is the unique minimal degree of some polynomial $|x|^k+k$ such that $M(x)$ will halt within $|x|^k+k$ steps for any input $x$ of length $|x|$. We call such a $k$ the order of $(M,k)$.

To obtain our main result, we need to {\it enumerate} the polynomial-time nondeterministic Turing machines so that for each positive integer $i$, there is a unique tuple of $(M,k)$ associated with $i$ (i.e., to define a surjective function from $\mathbb{N}_1$ to the set of all polynomial-time nondeterministic Turing machines $\{(M,k)\}$), such that we can refer to the $j$-th polynomial-time nondeterministic Turing machine.

By Lemma \ref{lemma3}, we can restrict ourselves to single-tape nondeterministic Turing machines. So, in the following, by polynomial-time nondeterministic Turing machines we mean single-tape polynomial-time nondeterministic Turing machines.

We first present the following important definition concerning ``enumeration" of a set $T$ which is enumerable: \footnote{ In Georg Cantor's terminology \cite{Can91}, enumeration of something is the ``sequence" of something.}

 \begin{definition}[\cite{Rud76}, p. 27, Definition 2.7]\footnote{ There exist enumerations for $T$ when and only when the set $T$ is enumerable. And the term ``{\em enumerable}", Turing refers to \cite{Hob21}, p. 78. That is, that the set $T$ is enumerable is the same as that $T$ is countable. See \cite{Tur37}, Section of {\em Enumeration of computable sequences}.}
 \label{definition4}
 By an enumeration of set $T$, we mean a function $e$ defined on the set $\mathbb{N}_1$ of all positive integers. If $e(n)=x_n\in T$, for $n\in\mathbb{N}_1$, it is customary to denote the enumeration $e$ by the symbol $\{x_n\}$, or sometimes by $x_1$, $x_2$, $x_3$, $\cdots$. The values of $e$, that is, the elements $x_n\in T$, are called the {\em terms} of the enumeration.
 \end{definition}

To show that the set of all polynomial-time nondeterministic Turing machines is enumerable (and to further present an enumeration of the set of all polynomial-time nondeterministic Turing machines), we first use the method presented in \cite{AHU74}, page $407$, to encode a single-tape nondeterministic Turing machine into an integer.

Without loss of generality, we can make the following assumptions about the representation of a single-tape nondeterministic Turing machine with input alphabet $\{0,1\}$, because that will be all we need: 

\begin{enumerate}
\item {The states are named 
$$
q_1,q_2,\cdots,q_s
$$
for some $s$, with $q_1$ the initial state and $q_s$ the accepting state.}
\item {The input alphabet is $\{0,1\}$.}
\item {The tape alphabet is 
$$
\{X_1,X_2,\cdots,X_t\}
$$
for some $t$, where $X_1=\mathrm{b}$, $X_2=0$, and $X_3=1$.}
\item {The next-move function $\delta$ is a list of quintuples of the form, $$
    \{(q_i,X_j,q_k,X_l,D_m),\cdots,(q_i,X_j,q_f,X_p,D_m)\}
     $$
     meaning that 
     $$
     \delta(q_i,X_j)=\{(q_k,X_l,D_m),\cdots,(q_f,X_p,D_m)\},
      $$
      and $D_m$ is the direction, $L$, $R$, or $S$, if $m=1,2$, or $3$, respectively. We assume this quintuple is encoded by the string 
      $$
      10^i10^j10^k10^l10^m1\cdots 10^i10^j10^f10^p10^m1.
      $$
      }
\item {The nondeterministic Turing machine itself is encoded by concatenating in any order the codes for each of the quintuples in its next-move function. Additional $1$'s may be prefixed to the string if desired. The result will be some string of $0$'s and $1$'s, beginning with $1$, which we can interpret as an integer.}
\end{enumerate}

Next, we encode the order of $(M,k)$ as
    $$
    10^k1
    $$
so that the tuple $(M,k)$ is the concatenation of the binary string representing $M$ itself followed by the order $10^k1$. Now the tuple $(M,k)$ is encoded as a binary string, which can be interpreted as an integer.

By this encoding, any integer that cannot be decoded into a polynomial-time nondeterministic Turing machine is assumed to represent the trivial Turing machine with an empty next-move function. Every single-tape polynomial-time nondeterministic Turing machine will appear infinitely often in the enumeration since, given a polynomial-time nondeterministic Turing machine, we may prefix $1$'s at will to find larger and larger integers representing the same set of $(M,k)$. We denote such a polynomial-time NTM by $\widehat{M}_j$, where $j$ is the integer representing the tuple $(M,k)$.
    
\vskip 0.3 cm
\begin{remark}
\label{remark3.1}
One of the conveniences of representing the polynomial-time nondeterministic Turing machines in this way (i.e., as a tuple $(M,k)$) is to conveniently control the running space of the universal nondeterministic Turing machine $M_0$ constructed in Theorem \ref{theorem3} in subsection \ref{sec:proofoftheorem1} below, so that it facilitates our analysis of the space complexity of $M_0$, i.e., to easily establish the fact of Theorem \ref{theorem4}.
\end{remark}
\vskip 0.3 cm

In fact, we have presented a function $e$ from $\mathbb{N}_1$ to the set of all polynomial-time nondeterministic Turing machines in the above, which has the property that there is only one 
$$
i\in\mathbb{N}_1
$$
for each polynomial-time nondeterministic Turing machine $(M,k)$ such that 
$$
e(i) = (M, k).
$$
Or equivalently, we have shown that the set of all polynomial-time nondeterministic Turing machines is enumerable, and 
$$
e:\mathbb{N}_1\rightarrow\{(M,k)\}
$$
is an enumeration. To conclude, we have the following: 
\begin{theoremsection}
\label{theorem3.1}
All of the polynomial-time nondeterministic Turing machines are in the above enumeration $e$. In other words, the set $\{(M,k)\}$ of all polynomial-time nondeterministic Turing machines is enumerable. \Q.E.D
\end{theoremsection}

\vskip 0.3 cm
\begin{remark}
\label{remark2}
There is another way to {\em enumerate} all polynomial-time nondeterministic Turing machines without encoding the order of the polynomial into their representation. To do so, we need the {\em Cantor pairing function} (see Fig. \ref{2} below from \cite{ANMOUS4}):
$$
\pi:\mathbb{N}\times\mathbb{N}\rightarrow\mathbb{N}
$$
defined by
$$
\pi(k_1,k_2):=\frac{1}{2}(k_1+k_2)(k_1+k_2+1)+k_2
$$
where $k_1,k_2\in\mathbb{N}$. Since the Cantor pairing function is invertible (see \cite{ANMOUS4}), it is a bijection between $\mathbb{N}\times\mathbb{N}$ and $\mathbb{N}$. As we have shown that any polynomial-time nondeterministic Turing machine is an integer, we can place any polynomial-time nondeterministic Turing machine and its order of polynomial into a tuple $(M,k)$ and use the Cantor pairing function to map the tuple $(M,k)$ to an integer in $\mathbb{N}_1$. The reader can easily check that such a method gives an enumeration of the set of polynomial-time nondeterministic Turing machines.

\vskip 0.3 cm   
\begin{figure}[htb]
     \center{\includegraphics[width=9.5cm]{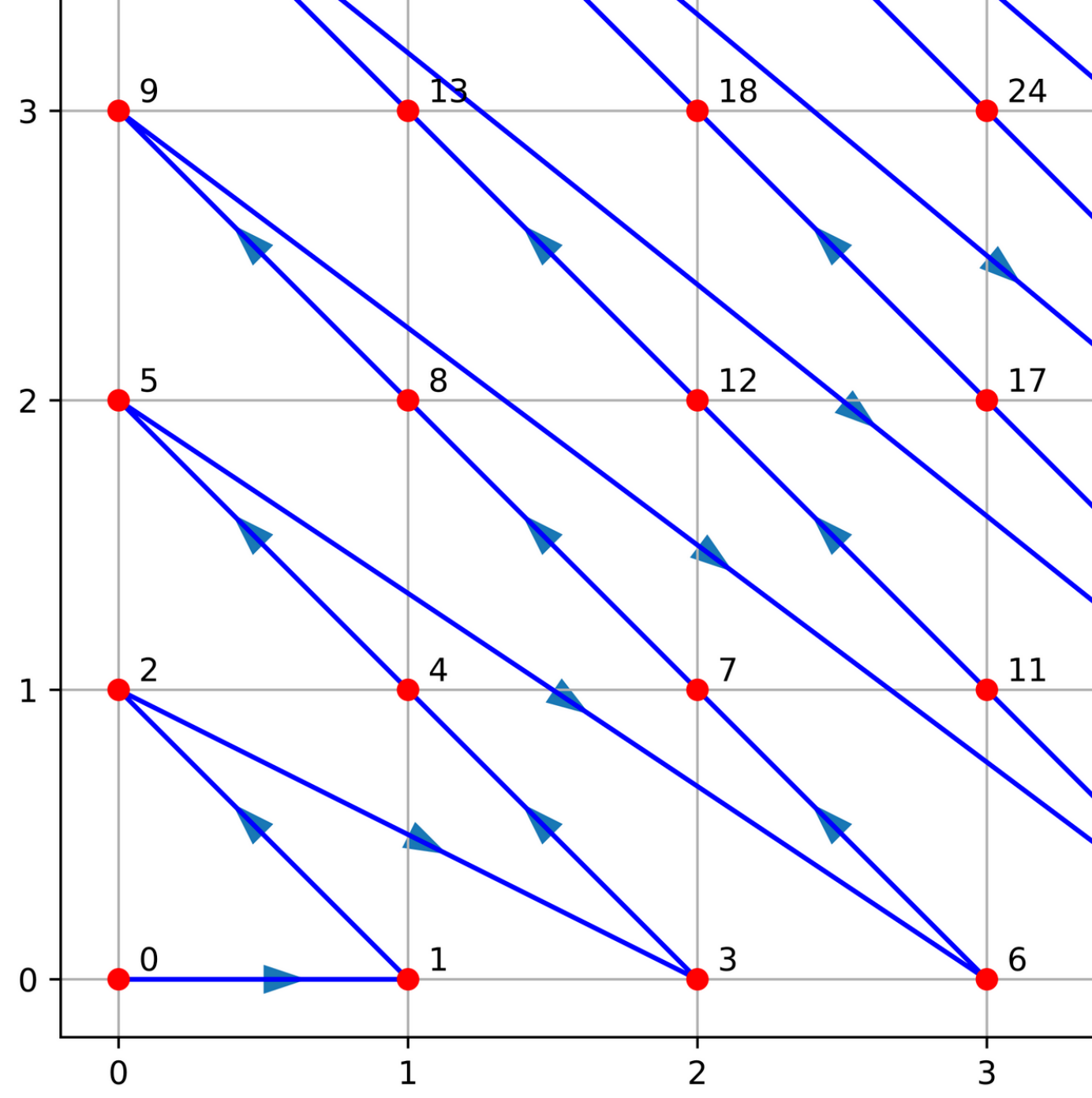}}
     \caption{\label{2}Cantor pairing function}
  \end{figure}
\end{remark}

\vskip 0.3cm
\section{$\mathcal{NP}$ Differs from $\mathcal{PSPACE}$}
\label{sec:proofofmainresult}
 \vskip 0.3cm
 
Note that we intend to demonstrate 
$$
\mathcal{NP}\subsetneqq \mathcal{PSPACE}.
$$

\vskip 0.5cm
\subsection{Proof of Theorem \ref{theorem1}}
\label{sec:proofoftheorem1}

We are now able to design a four-tape NTM $M_0$ that treats its input string $x$ both as an encoding of a tuple $(M_j,k)$ and also as the input to the polynomial-time nondeterministic Turing machine $M_j$. We stress that the deterministic simulation process in the following important theorem is adapting from the simulation process in \cite{DK14} (see the proof of Theorem 1.9 in \cite{DK14}):

\begin{theoremsection}
\label{theorem3}
There exists a language $L_d$ accepted by a universal nondeterministic Turing machine $M_0$ but by no polynomial-time nondeterministic Turing machines (i.e., $L_d\notin\mathcal{NP}$).
\end{theoremsection}

\begin{proof}
Let $M_0$ be a four-tape NTM which operates as follows on an input string $x$ of length of $n$.
\begin{enumerate}

\item{
$M_0$ places its input $x$ on tape $1$ and the context is never changed during the computation.
}

\item{
By using 
$$
O(\log |x|)
$$
space,\footnote{ In this paper, $\log n$ stands for $\log_2n$.} $M_0$ decodes the tuple encoded by $x$. If $x$ is not the encoding of some non-trivial polynomial-time NTM $M_j$ for some $j$, then GOTO $7$, else determines $t$, the number of tape symbols used by $M_j$; $s$, its number of states; and $k$, its order of the polynomial 
$$
T(n)=n^k+k,
$$
meaning that $M_j$ is a nondeterministic $n^k+k$ time-bounded Turing machine and its time complexity is of order $k$. In particular, $M_0$ determines the positive integer $l$ such that 
$$
l=\max\{|\delta(q,a)|\,:\,\text{ for all $q\in Q$ and for all $a\in\Gamma$ of $M_j$}\}
$$
Note that $l$ is finite, i.e., $l<+\infty$. The third tape of $M_0$ can be used as ``scratch" memory to calculate $t$.}

\item{
$M_0$ marks off 
$$
|x|^{k+1}
$$
cells on each tape except tape $1$ (note that tape $1$ is the input tape of $M_0$). After doing so, if any tape head of $M_0$ attempts to move off the marked cells, $M_0$ halts without accepting.
}

\item{ $M_0$ lays off on its second tape 
$$
T(|x|)=|x|^k+k
$$
blocks of 
$$
\lceil\log t\rceil
$$
cells each, the blocks being separated by single cell holding a marker $\#$, i.e., there are 
$$
(1+\lceil\log t\rceil)T(|x|)
$$
cells in all. Each tape symbol occurring in a cell of $M_j$'s tape will be encoded as a binary number in the corresponding block of the second tape of $M_0$. Initially, $M_0$ places $M_j$'s input, in binary-coded form, in the blocks of tape $2$, filling the unused blocks with the code for the blank.
}

\item{
 $M_0$ then lays off on its $3$th tape 
$$
T(|x|)=|x|^k+k
$$
blocks of 
$$
\lceil\log l\rceil
$$
cells each, i.e., there are 
$$
(\lceil\log l\rceil)T(|x|)
$$
cells in all. $M_0$ generates the $r$th string $w$ in $\{1,2,\cdots,l\}^{T(|x|)}$, in the lexicographic ordering, on tape $3$. (This string will overwrite the $(r-1)$th string generated in stage $r-1$.)
}

\item{
Next, $M_0$ on input $x$ simulates $M_j$ on input $x$ deterministically. Specifically, $M_0$ simulates $M_j$, using tape $1$, its input tape, to determine the moves of $M_j$ and using tape $2$ to simulate the tape of $M_j$. $M_0$ uses tape $3$ to store the current simulation information and tape $4$ is used to hold the state of $M_j$. More precisely, $M_0$ operates on input $x$ in stages. At stage $r>0$, $M_0$ performs the following actions:

\begin{enumerate}
  \item [(1)]{ $M_0$ erases anything in tape $2$ that may have been left over from stage $r-1$ and copies the input $x$ from tape $1$ to tape $2$, i.e., $M_0$ places $M_j$'s input, in binary-coded form, in the blocks of tape $2$, filling the unused blocks with the code for the blank. Note that each tape symbol occurring in a cell of $M_j$'s tape will be encoded as a binary number in the corresponding block of the second tape of $M_0$.}
  \item [(2)]{$M_0$ generates the $r$th string $w_1w_2\cdots w_{T(|x|)}\in\{1,2,\cdots,l\}^{T(|x|)}$, in the lexicographic ordering, on tape $3$.}
  \item [(3)]{$M_0$ simulates $M_j$ on input $x$ on tape $2$ for at most $T(|x|)$ moves. At the $j$th move, $1\leq j\leq T(|x|)$, $M_0$ examines the $j$th symbol $w_j$ on the tape $3$ to determine which transition of the relation $\delta$ is to be simulated. More precisely, if the current state is $q$ and the symbol currently scanned on tape $2$ is $a$, and $\delta(q,a)$ contains at least $l_j$ values and $w_j\leq l_j$, then $M_0$ follows the $w_j$-th move of $\delta(q,a)$ (note that $w_j\in\{1,2,\cdots,l\}$ is a positive integer); if $\delta(q,a)$ has less than $w_j$ values, then $M_0$ goes to stage $r+1$.}
  \item [(4)]{If the simulation halts within $T(|x|)$ moves in an accepting state, then $M_0$ rejects the input $x$ and halts; Else if the simulation halts within $T(|x|)$ moves in an rejecting state
  and the context $w$ of tape $3$ is 
  $$\aligned
  w=&w_1w_2\cdots w_{T(|x|)}\\
   =&\underbrace{ll\cdots l}_{T(|x|)},
  \endaligned$$
   then $M_0$ accepts the input $x$ and halts; otherwise, it goes to stage $r+1$.
  }
\end{enumerate}
}\label{item4}

\item{
Since $x$ is not encoding of some non-trivial polynomial-time NTM, $M_0$ marks off 
$$
(1+\lceil \log 4\rceil+\lceil\log 2\rceil)|x|
$$
on each tape except tape $1$. After doing so, if any tape head of $M_0$ attempts to move off the marked cells, $M_0$ halts without accepting. Then, on tape $3$, $M_0$ sets up a block of 
$$
\lceil\log 4\rceil+\lceil\log |x|\rceil +\lceil\log 2\rceil |x|
$$ 
cells, initialized to all $0$'s. Tape $3$ is used as a counter to count up to 
$$
4\times |x|\times 2^{|x|}.
$$
By using its nondeterministic choices, $M_0$ moves according to the path given by $x$. The moves of $M_0$ are counted in binary in the block of tape $3$. $M_0$ rejects if the number of moves made by $M_0$ exceeds 
$$
4\times|x|\times 2^{|x|}
$$ 
or $M_0$ reaches a reject state before reaching 
$$
4\times|x|\times 2^{|x|},
$$
otherwise $M_0$ accepts. Note that the factors of $4$ and $2$ in 
$$
4\times|x|\times 2^{|x|}
$$
is fixed; they are defaults.
}
\end{enumerate}

The NTM $M_0$ described above is of space complexity $C(n)$ which is currently unknown, and  of course it accepts some language $L_d$.

Suppose now $L_d$ were accepted by some polynomial-time NTM $M_i$ in the enumeration $e$ whose time complexity is 
$$
T(n)=n^k+k.
$$
Then let $M_i$ have $s$ states and $t$ tape symbols. 

Further, let 
$$
m=\max\{\lceil\log l\rceil,1+\lceil \log t\rceil\}.
$$
Then,
$$
m\times T(|x|)=\max\{(\lceil\log l\rceil)\times T(|x|),(1+\lceil \log t\rceil)\times T(|x|)\}.
$$

\noindent Since

$$\aligned
\lim_{n\rightarrow\infty}&\frac{\max\{(\lceil\log l\rceil)\times T(n),(1+\lceil \log t\rceil)\times T(n)\}}{n^{k+1}}\\
       =&\lim_{n\rightarrow\infty}\frac{mT(n)}{n^{k+1}}\\
       =&\lim_{n\rightarrow\infty}\frac{m\times (n^k+k)}{n^{k+1}}\\
       =&\lim_{n\rightarrow\infty}\left(\frac{m\times n^k}{n^{k+1}}+\frac{m\times k}{n^{k+1}}\right)\\
       =&0\\
       <&1,
\endaligned$$

\noindent there exists an $N_0>0$ such that for any $N\geq N_0$,
$$
\max\{(\lceil\log l\rceil)\times T(N),(1+\lceil \log t\rceil)\times T(N)\}<N^{k+1},
$$
which implies that for a sufficiently long $x$, say $|x|\geq N_0$, and $M_x$ denoted by such $x$ is $\widehat{M}_i$, we have
$$
\max\{(\lceil\log l\rceil)\times T(|x|),(1+\lceil \log t\rceil)\times T(|x|)\}<|x|^{k+1}.
$$

Thus, on input $x$, $M_0$ has sufficient room to simulate $\widehat{M}_i$ and accepts if and only if $\widehat{M}_i$ rejects, a contradiction with our assumption that $M_i$ accepted $L_d$, i.e., $M_i$ agreed with $M_0$ on all inputs.

The above arguments further yield that there exists no polynomial-time nondeterministic Turing machine $M_j$ in the enumeration $e$ accepting the language $L_d$. Since all polynomial-time nondeterministic Turing machines are in the list of $e$, thus, 
$$
L_d\notin\mathcal{NP}.
$$
The proof is completed.
\end{proof}

\vskip 0.3 cm
\begin{remark}
The simulation techniques for a polynomial-time nondeterministic Turing machine in the proof of Theorem \ref{theorem3} are essentially the same as in the proof of Lemma \ref{lemma1}. We just present a concretization of the simulation process in the proof of Lemma \ref{lemma1}. 
\end{remark}

Next, we will show that the universal nondeterministic Turing machine $M_0$ works in space $O(n^k)$ for any $k\in\mathbb{N}_1$:

\begin{theoremsection}
\label{theorem4}
The universal nondeterministic Turing machine $M_0$ constructed in proof of Theorem \ref{theorem3} runs within space $O(n^k)$ for any $k\in\mathbb{N}_1$.
\end{theoremsection}

\begin{proof}
The simplest way to show the theorem is to prove that for any input $w$ to $M_0$, there is a corresponding positive integer $i_w\in\mathbb{N}_1$ such that $M_0$ scans at most $|w|^{i_w+1}$ cells, which can be done as follows.
   
On the one hand, when the input $x$ encodes a polynomial-time nondeterministic Turing machine $M_j$, say a nondeterministic $T(n)$ time-bounded Turing machine, where
$$
T(n) = n^k+k,
$$
then $M_0$ scans at most 
$$
|x|^{k+1}
$$
cells by the construction. So, the corresponding integer is $k$ (i.e., $i=k$) in this case, and $M_0$ works in space $O(n^k)$. This holds for any nondeterministic $n^k+k$ time-bounded Turing machine in the enumeration $e$ where $k\in\mathbb{N}_1$.

On the other hand, when the input $x$ encodes a trivial Turing machine, then there is a constant $c>0$, such that $M_0$ scans no more than 
$$
c|x|
$$
cells for the input $x$. So the corresponding integer is $1$ (i.e., $i_x=1$, $M_0$ scans at most $n^2$ cells) in this case, and $M_0$ works in space $O(n)$. In both cases, $M_0$ runs in space at most :
$$\aligned
C(n) = & \max\{n^k, n\}\\
       = & n^k
\endaligned$$

\noindent for any $k\in\mathbb{N}_1$.
\end{proof}

 \vskip 0.3 cm
\begin{remark}
 \label{remark3}
Some may argue that $M_0$ runs within space $O(n^k)$ for any $k\in\mathbb{N}_1$, thus it is not obvious that $L_d\in\mathcal{PSPACE}$. We will  first  prove rigorously that $L_d\in\mathcal{NSPACE}$.
\end{remark}
\vskip 0.3 cm

\begin{theoremsection}
\label{theorem5}
The language $L_d$ accepted by $M_0$ is in $\mathcal{NSPACE}$.
\end{theoremsection}
\begin{proof}
We first define the family of languages 
$$
\left\{L_d^i\right\}_{i\in\mathbb{N}_1}
$$
as follows:\footnote{This can be done as follows: we insert item $3'$ after item 2 and before item 3 in the proof of Theorem \ref{theorem3}. Namely, 
\begin{enumerate}
  \item [$3'$.]{$M_0$ marks off $|x|^{i+1}$ cells on each tape except tape $1$. After doing so, if any work-tape heads of $M_0$ attempts to move off the marked cells, $M_0$ halts without accepting.}
\end{enumerate}
Then, in item 3, $M_0$ marks off $|x|^{k+1}$ cells within the $|x|^{i+1}$ cells previously marked. If $k+1>i+1$, $M_0$ will fail to mark off $|x|^{k+1}$ cells, thus rejecting the input. The above revision, in fact, limits the total cells for $M_0$ to work to at most $n^{i+1}$. 

Equivalently, first let $M_0$ compare the order $k$ of the nondeterministic $n^k+k$ time-bounded Turing machine $M_j$ decoded from the input $x$ and the value $i$ (a fixed constant); if $k+1>i+1$
then rejects the input $x$ (since in this case, when simulating the nondeterministic $n^k+k$ time-bounded Turing machine $M_j$, $M_0$ will scan the cells beyond $n^{i+1}$), else $M_0$ marks off $|x|^{k+1}$ cells on each work tape (in this case, when simulating the nondeterministic $n^k+k$ time-bounded Turing machine $M_j$, $M_0$ scans cells within $n^{i+1}$), which is slightly different from \cite{Lin21b}.}
\begin{align*}
L_d^i\overset{\text{def}}{=}\,\,&\text{the language accepted by $M_0$ scanning $O(n^i)$ cells for fixed $i\in\mathbb{N}_1$ (i.e., at }\\
&\text{most $n^{i+1}$ cells for fixed $i\in\mathbb{N}_1$). That is, $M_0$ if forced to halt}\\
&\text{when its work-tape heads attempt to move off the marked $n^{i+1}$ cells }\\
&\text{during the computation.}
\end{align*}

Then, by construction and by Theorem \ref{theorem4}, for each input $w$ to $M_0$ there exists a corresponding integer $i_w$ such that $M_0$ scans at most $|w|^{i_w+1}$ cells (i.e., $M_0$ scans at most $|w|^{i_w+1}$ cells for each work tape, or runs within space $O(n^i)$ for any $i\in\mathbb{N}_1$); we thus have that
$$
 L_d=\bigcup_{i\in\mathbb{N}_1}L_d^i.\eqno(3.2.1)
$$

Moreover,
$$
L_d^i\subseteq L_d^{i+1},\quad\text{for each fixed $i\in\mathbb{N}_1$}
$$
since any word $w\in L_d^i$ accepted by $M_0$ within space $O(n^i)$ can certainly be accepted by $M_0$ within space $O(n^{i+1})$, i.e.,
$$
w\in L_d^{i+1}.
$$
This gives that for any $i\in\mathbb{N}_1$,
$$
         L_d^1\subseteq L_d^2\subseteq\cdots\subseteq L_d^i\subseteq L_d^{i+1}\subseteq\cdots\eqno(3.2.2)
$$

Note further that $L_d^i$ is accepted by a multi-tape nondeterministic Turing machine $M_0$ within space $O(n^i)$, by Lemma \ref{lemma4}, we thus obtain that $L_d^i$ is accepted by a single-tape nondeterministic Turing machine $M_0'$ within space 
$$
O(n^i).
$$
Hence
$$
L_d^i\in\text{NSPACE}[n^i]\subseteq\mathcal{NSPACE},\quad\text{for any fixed $i\in\mathbb{N}_1$.}\eqno(3.2.3)
$$

It is clear that (3.2.1), together with (3.2.2) and (3.2.3), yields
$$
L_d \in\mathcal{NSPACE}.
$$This completes the proof.
\end{proof}
   
\vskip 0.3 cm
\begin{remark}
In fact, after obtaining the relations {\em (3.2.1)} and {\em (3.2.2)}, we can suppose that 
$$
L_d\not\in\mathcal{NPSPACE}.
$$
Then there must exist at least one $i\in\mathbb{N}_1$ such that 
$$
L_d^i\not\in\mathcal{NPSPACE}.
$$
But by definition, $L_d^i$ is the language accepted by the nondeterministic Turing machine $M_0$ scanning at most $n^{i+1}$ cells (i.e., working within $O(n^i)$ space), which is clearly a contradiction. Hence such an $i$ cannot be found, i.e., 
$$
L_d^i\in\mathcal{NSPACE}
$$
for all $i\in\mathbb{N}_1$. Equivalently, 
$$
L_d\in\mathcal{NSPACE}.
$$
\end{remark}
\vskip 0.3cm
   
Now we are in a position to prove Theorem \ref{theorem1}:

\vskip 0.3cm

\noindent{\it Proof of Theorem \ref{theorem1}}. First note that by Savitch's Theorem \cite{Sav70}, we have 
$$
\mathcal{PSPACE}=\mathcal{NSPACE}.
$$
Then it is obvious that Theorem \ref{theorem1} is an immediate consequence of Theorem \ref{theorem3} and Theorem \ref{theorem5}. \Q.E.D

\vskip 0.3 cm
\begin{remark}
Note that we have shown in Theorem \ref{theorem5} that $L_d\in\mathcal{NSPACE}$ mathematically, but some readers may have the following question: can we find a fixed constant 
$$
c\in\mathbb{N}_1
$$
such that the nondeterministic Turing machine $M_0$ runs within polynomial space
$$
n^c+c\,?
$$
The answer depends on whether we can answer the following question. Let 
$$
NPTMs=\{T_1,T_2,\cdots\}
$$
be the set of all polynomial-time nondeterministic Turing machines, and let 
$$
{\rm order}(T_i)
$$
be the order of polynomial of machine $T_i$. For example, if $T_i$ is a nondeterministic $n^z+z$ time-bounded Turing machine, then $$
{\rm order}(T_i)=z.
$$
Let 
$$
m=\max\left\{{\rm order}(T_1),{\rm order}(T_2),\cdots\right\}
$$
Then we can say that
$$
n^{m+1}+(m+1)
$$
is the polynomial of $M_0$. But can we find such a fixed constant $c$ in $\mathbb{N}_1$ so that 
$$
c=m\,?
$$
\end{remark}
\vskip 0.3cm
  
\section{Conclusions and Open Problems}
\label{sec:conclusions}

In conclusion, we have shown that 
$$
\mathcal{NP}\subsetneqq\mathcal{PSPACE}.
$$

The essential techniques used are basically the same as those in the author's \cite{Lin21b}. In other words, we enumerate all polynomial-time nondeterministic Turing machines, then diagonalize against all of the polynomial-time nondeterministic Turing machines in the enumeration $e$. The language accepted by $M_0$, which differs from each polynomial-time nondeterministic Turing machine in the enumeration $e$ and belongs to $\mathcal{NSPACE}$, yields our main result. The last step is by Savitch's Theorem \cite{Sav70}. 

The main practical implication is that $\mathcal{PSPACE}$-complete problems are indeed harder than $\mathcal{NP}$-complete problems.

Obviously, we did not touch on the question of whether 
$$
\mathit{NL}\overset{?}{=}\mathcal{P}
$$
in this paper. Moreover, the following is a very important open problem: how to prove or disprove that $$\mathit{PH}\neq\mathcal{PSPACE}?$$where $\mathit{PH}$ is the polynomial hierarchy \cite{ANMOUS5}. Furthermore, there are many other interesting open questions for future study; see e.g., \cite{ANMOUS1}. 

Although, it was showed in \cite{BBBV97} that relative to an oracle chosen uniformly at random with probability $1$ the class $\mathcal{NP}$ cannot be solved on a quantum Turing machine in time $o(2^{\frac{n}{2}})$, this does not necessarily imply $\mathcal{NP}\not\subset\mathcal{BQP}$, where $\mathcal{BQP}$ is the complexity class of bounded-error quantum polynomial-time \cite{BV97}, because the oracle result is not a necessary and sufficient condition for $\mathcal{NP}\not\subset\mathcal{BQP}$, which means that the exact relationship between $\mathcal{BQP}$ and $\mathcal{NP}$ is unknown (i.e., we cannot deduce that $\mathcal{NP}$-complete problems are outside the class $\mathcal{BQP}$, although it is a general belief that the $\mathcal{NP}$-complete problems are outside the class $\mathcal{BQP}$). This is therefore a very fundamental and important open question --- one that will infect the security of the McEliece cryptosystem \cite{Mce78} --- based on the fact that decoding an arbitrary linear code is $\mathcal{NP}$-complete \cite{BMT78}.

\bibliographystyle{ACM-Reference-Format}

\end{document}